\documentclass[conference, twocolumns, a4]{IEEEtran}

\usepackage[T1]{fontenc}
\usepackage[english]{babel}
\usepackage[utf8]{inputenc}

\usepackage{breakcites}
\usepackage[hyphenbreaks]{breakurl}
\usepackage{acronym}
\usepackage{graphicx}
\usepackage{epstopdf}
\usepackage{amsthm}
\usepackage{amsmath}
\usepackage{amssymb}
\usepackage{tabularx}
\usepackage[algoruled, vlined]{algorithm2e}
\usepackage{enumitem}
\usepackage{comment}
\setlist[description]{leftmargin=0cm,labelindent=0cm}
\usepackage{subfig}
\usepackage{setspace}
\usepackage{xcolor}

\newtheorem{theorem}{Theorem}

\newtheorem{proposition}[theorem]{Proposition}

\newcommand{\suml}{\sum\limits}

\begin{document}

\setlength{\textfloatsep}{10pt}
\setlength{\intextsep}{10pt}

%
\title{Overlay Routing for Fast Video Transfers in CDN}

\author{\IEEEauthorblockN{Paolo Medagliani\IEEEauthorrefmark{1}, Stefano Paris\IEEEauthorrefmark{1}, Jérémie Leguay\IEEEauthorrefmark{1}, Lorenzo Maggi\IEEEauthorrefmark{1}, Xue Chuangsong\IEEEauthorrefmark{2}, Haojun Zhou\IEEEauthorrefmark{2}}
	\IEEEauthorblockA{\IEEEauthorrefmark{1}Mathematical and Algorithmic Sciences Lab, France Research Center, Huawei Technologies Co. Ltd.\\
		20 Quai du Point du Jour, 92100 Boulogne-Billancourt, France\\
}
	\IEEEauthorblockA{\IEEEauthorrefmark{2}Carrier Software Unit, Huawei Technologies Co. Ltd., Nanjing, China\\
	}

}

\maketitle

\begin{abstract} 
Content Delivery Networks (CDN) are witnessing the outburst of video streaming (e.g., personal live streaming or Video-on-Demand) where the video content, produced or accessed by mobile phones, must be quickly transferred from a point to another of the network. Whenever a user requests a video not directly available at the edge server, the CDN network must 1) identify the best location in the network where the content is stored, 2) set up a connection and 3) deliver the video as quickly as possible. For this reason, existing CDNs are adopting an overlay structure to reduce latency, leveraging the flexibility introduced by the Software Defined Networking (SDN) paradigm. In order to guarantee a satisfactory Quality of Experience (QoE) to users, the connection must respect several Quality of Service (QoS) constraints. In this paper, we focus on the sub-problem 2), by presenting an approach to efficiently compute and maintain paths in the overlay network. Our approach allows to speed up the transfer of video segments by finding minimum delay overlay paths under constraints on hop count, jitter, packet loss and relay processing capacity. The proposed algorithm provides a near-optimal solution, while drastically reducing the execution time. We show on traces collected in a real CDN that our solution allows to maximize the number of fast video transfers.
\end{abstract}

\IEEEpeerreviewmaketitle

\section{Introduction}

Globally, IP video traffic is expected to represent 82 percent of all IP traffic (business and consumer) by 2020~\cite{intro_Cisco_VNI}. Internet video traffic is expected to grow fourfold from 2015 to 2020. While a large variety of Video on Demand (VoD) and video-streaming services has emerged in the past years, the field continues to evolve rapidly. The ways that people are watching video is constantly evolving and is driven by mobile usage. For instance, live streaming embedded in social media platforms is a relatively new phenomenon, but this technology is finding more and more support with services such as Facebook Live or Periscope. 

With the explosion of streaming services that deliver Internet video to the TV and other device endpoints, Content Delivery Networks (CDN) have prevailed as a dominant method to deliver such content. Globally, 72 percent of Internet video traffic will cross CDN by 2019. The largest over-the-top player Akamai currently has over 170,000 edge servers located in over 1300 networks in 102 countries~\cite{maggs2015algorithmic}. At a smaller scale, Internet service providers are also deploying their own infrastructure, referred to as Telco CDN, as an evolution of IPTV and VoD systems. CDN have been traditionally used to help content providers distributing static content at scale. They are typically composed of edge servers which are deployed as close as possible to end users and act as a proximity cache. However, to follow the evolution of usage towards  more dynamic and real-time services, CDN are evolving to support a large variety of content types which cannot always be cached, such as web applications, teleconferencing and live video streaming. 

Delivering content at scale over the Internet with latency and reliability constraints is a real challenge. Indeed, the Internet is best-effort with routing policies that do not address fined-grained needs of applications and that are often guided by business relationships on large traffic volumes. The Triangle Inequality Violation (TIV)~\cite{Lumezanu:2009}  is a well known consequence of such policies. The minimum delay path is almost never the one established by the underlying routing system. In addition, outages are happening all the time in the Internet due to cable cuts, misconfigurated routers, DDoS attacks, power outages, or natural disasters~\cite{nygren2010akamai}. Even if the Internet becomes flatter~\cite{chiu2015we} with content service providers buying direct connectivity closer to their end users, CDN operators are still fighting against TIV and best effort routing policies. In reaction, overlay networks, such as RON~\cite{andersen2002resilient}, have been introduced to provide low latency and reliable connectivity over the Internet. Similarly, CDN operators deploy a three-tier architecture composed of origin servers that create the content, edge servers, which clients access to consume the content, and an overlay network that is responsible for transporting the content from the origins to the edges. Clients request the content from the closest edge server, and the edge server in turn retrieves the requested content from the origin via the overlay network over the Internet. 

CDN solutions are composed of building blocks such as a caching system to store the most popular contents at the network edge, a load balancing system, integrated within a Domain Name System (DNS) server, to redirect client requests to the closest edge server and an overlay routing system to transport content at low latency and high reliability. The overlay routing system is invoked to find \emph{good} paths at a number of occasions. This is for instance required to connect a live streaming content producer to its consumers, or to retrieve segments of a non real-time video stream which are cached at a given edge sever. Following the ongoing transformation of network architectures with Software Defined Networks (SDN)~\cite{Nunes14}, CDN are adopting flow-oriented and centralized controllers~\cite{liu2012case} to manage video traffic especially. The (logically) centralized control aims at improving the Quality of Experience (QoE) perceived by end users. The challenge for such an overlay network controller is to quickly find paths in the overlay when new demands arrive and to maintain a good routing configuration over time so that the transfer time of video segments is minimized. 

This paper presents an efficient algorithm to maintain minimum delay overlay paths with multiple QoS constraints on capacity, jitter and packet loss. These optimization criteria have been carefully selected to speed-up the transfer of video segments, knowing that TCP or QUIC~\cite{quic} is used under HLS or MPEG-DASH~\cite{seufert2015survey} for live and regular streaming. While the optimization problem it solves is NP-Hard, we use tools from combinatorial optimization such as Lagrangian relaxation and column generation to quickly find a near-optimal approximation. Our algorithm is based on the decomposition of the original problem into two simpler problems, namely finding a  minimum delay path that satisfies all QoS constraints for a single demand (QoS path computation), and computing the best set of paths for all demands (constrained multi-commodity flow). We present an evaluation of this algorithm over measurements collected in a real Telco CDN. We extend the evaluation over a synthetic network to further highlight the performance of the algorithm. 

The rest of the paper is structured as follows.
Sec.~\ref{related} provides an overview of the related work. Sec.~\ref{formulation} introduces the system model considered and formulates the problem as an MILP. Sec.~\ref{algorithms} describes the algorithms that we propose to solve the admission and routing maintenance problems.
Sec.~\ref{results} illustrates the numerical evaluation of the proposed algorithms. Finally, concluding remarks are provided in Sec.~\ref{conclusion}.


\section{Related works}
\label{related}

Overlay routing has received a lot of attention, especially in the domains of peer-to-peer, conferencing and CDNs.

QoS routing for multi-layer description video streams has been proposed to compute disjoint overlay paths for each sub-stream to increase reliability~\cite{begen2005multi}. Thi et al.~\cite{thi2015qos} allocate all sub-streams at once by minimizing an estimate of the end-to-end video quality distortion. They minimize the end-to-end probability of video stall based on packet loss and latency.

Distributed path computation algorithms have been proposed for peer-to-peer conferencing applications to build and maintain application layer multicast trees~\cite{Hosseini2007}. They gradually maintain a tree for each multicast session by defining join and leave procedures. Conversely, our work suits for the CDN case with a centralized controller platform~\cite{liu2012case}.

Andreev el al. \cite{Andreev11} introduce an overlay network for live streaming with edge servers, reflectors and source nodes. They propose a linear relaxation and rounding based algorithm to select reflectors. As reflectors can split multimedia streams to serve multiple receivers, the algorithm build an overlay forest to connect sources to receivers. To increase reliability, they also ensure that each source is served by two reflectors. Their primary objective is to optimize the cost of the infrastructure. Similarly, Zhou et al.~\cite{zhou2015joint} proposed an algorithm to find capacity and delay bounded minimum cost overlay forests. 

Our work differs from state of art by focusing on delay minimization subject to constraints on jitter and packet loss. In addition, it considers that edge servers act at the same time as source, relay and end points in the overlay network. In this context, we provide an algorithmic framework for the online and global control of fast video transfers in CDN.


\section{Problem formulation}
\label{formulation}

This section introduces the system model and the path computation problem.

\subsection{System model} \label{secIII_system}


\begin{figure}[!t]
	\centering
	\includegraphics[width=0.43\textwidth]{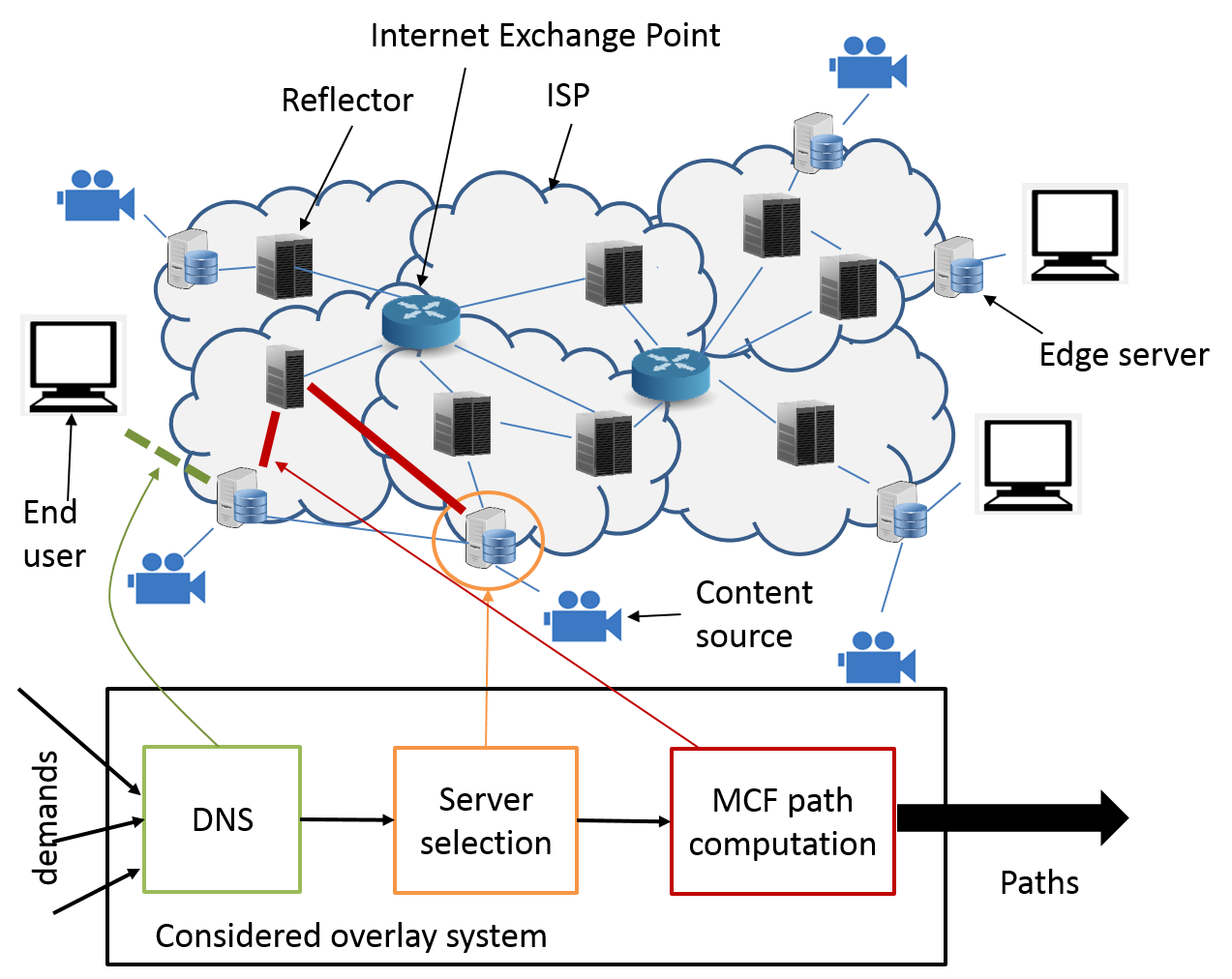}
	\caption{System model of a CDN overlay framework.}
	\label{fig:CDN_system}
\end{figure}

A typical CDN framework is depicted in Fig.~\ref{fig:CDN_system}.
We consider the presence of several content sources that stream videos, both live and on demand, to connected remote end users using the existing overlay network. Instead of connecting directly the end users to the source, end users connect to an edge server. 
In fact, since the edge server is able to replicate the same stream towards different end users, this results in a load reduction at sources and a better bandwidth utilization in the overlay. 
The choice of the best edge server which an end user must be connected to is typically handled by a DNS or HTTP proxy system. By leveraging different information such as geographical locations, content availability and edge server load, it redirects end user connection requests towards the most suitable edge server.
Once the best server has been identified, two cases arise: (i) the edge server has already the content available and it can serve the end user, (ii) the content is not available at the edge server. In the latter case, the edge server needs to retrieve the requested video from either the origin server or another edge server that already has the content available. 
To this end, it is necessary to compute the best overlay path to transfer the content to the edge server querying for it. Fig.~\ref{fig:CDN_system} shows a schematic description of the interactions among the different overlay components.

While existing overlays are composed of edge servers belonging to the same provider, we consider also the case where they can be extended with nodes placed at Internet eXchange Points (IXP) thanks the emerging technology of  Software-Defined Exchange (SDX). Fressancourt et al.~\cite{fressancourt2016kumori} have shown that IXP can be used as third-party routing inflection points to enhance an existing overlay network. This way, CDN overlays with even a few internal nodes can reach a high level of path diversity using external relays.

The system model presented above applies to two use cases: (i) video on demand (VoD) and (ii) personal live streaming (PLS). 
In the former scenario, end users request a video from a content provider. The goal of the system is to guarantee that the user is served as quickly as possible. If the content is already available at the edge server the problem is trivial and it only concerns the connection between end user and edge server. Instead, if the content is not directly available at the edge server, then the problem is equivalent to computing a path which minimizes the latency between two edge servers in the network and which respects some QoS constraints. The video content is either retrieved from the source or from the cache of another edge server.

In the PLS scenario, instead, content is generated and streamed by users, as it may happen for instance with Facebook Live. Other end users willing to watch the content connect to the overlay network in order to retrieve it. 
At this point, a similar situation to VoD arises. The overlay network first identifies the best edge server which the user must connect to; if the edge server does not have the content available, then it requests it from either the origin server or from another edge server~\cite{nygren2010akamai}. 
Once this choice has been made, the system will compute and maintain the fastest path between the content source and the demanding edge server. 
Throughout this paper, we consider that an external element has already chosen the best edge and origin servers from which the content must be retrieved. Hence, we will only focus on the path computation problem. 

In both cases, videos are streamed using HTTP Live Streaming or MPEG-DASH over TCP or QUIC. Minimizing transfer durations then translates into maximizing the throughput of each transport session. For TCP, the throughput can be approximated by the following formula $\frac{MSS.C}{RTT.\sqrt{p}}$~\cite{mathis1997macroscopic} which takes in to account the Maximum Segment Size (MSS), the Round-Trip Time (RTT) and the packet loss probability $p$. As a consequence, our path finding and maintenance algorithm aims at bounding packet loss and jitter while minimizing RTT. The parameters are continuously monitored by an active monitoring system. In addition, as the relaying of flows induces a burden on overlay nodes, we also consider a maximum amount of traffic that each one can process. 

We point out that the demands accepted by our algorithm will be routed in the CDN overlay, following the computed paths. For the refused demands, instead, they will be accepted anyway by the CDN overlay but using the direct path between origin and edge server in a best-effort way.

\subsection{Mathematical formulation}  \label{secIIIB_formulation}

In this paper we model the CDN overlay network as a weighted directed graph $G=(\mathcal N,\mathcal E)$, where $\mathcal N$ is the set of nodes and $\mathcal E$ denotes the set of edges. Each directed edge $(i,j)\in \mathcal E$ is characterized by its capacity $b_{i,j}$, delay $d_{i,j}$, jitter $z_{i,j}$ and successful packet transmission probability $f_{i,j}$. Each node $i\in \mathcal N$ has a maximum processing rate $N_i$.

\begin{table}[t!]
{\small
\centering
\begin{tabular}{| c | p{6.5cm} |}
	\hline
	\textbf{Symbol} & \textbf{Description} \\
	\hline
	$\mathcal{N}$ & Nodes (network devices). \\
	\hline
	$\mathcal{E}$ & Edges (network links). \\
	\hline 
	$\mathcal{K}$ & Set of demands (i.e., commodities). \\
	\hline
	$r^k$ & Transmission rate for demand $k\in\mathcal K$. \\
	\hline
	$d_{i_j}$ & Delay of edge $(i,j)\in \mathcal{E}$. \\
	\hline
	$b_{i,j}$ & Capacity of edge $(i,j)\in \mathcal{E}$. \\
	\hline
	$f_{i,j}$ & Prob. of successful transmission for edge $(i,j)$. \\
	\hline
	$F^k$ & Minimum prob. of successful transmission for $k$. \\
	\hline
	$z_{i,j}$ & Jitter of edge $(i,j)\in \mathcal{E}$. \\
	\hline
	$Z^k$ & Maximum jitter for demand $k\in\mathcal K$. \\
	\hline
	$N_i$ & Maximum processing rate for node $i\in\mathcal N$. \\
	\hline
\end{tabular}
}
\caption{Notations for input parameters.}
\label{tab:notation_input}
\end{table}

We consider a set $\mathcal{K}$ of video-streaming connection demands which need to be routed in the network. Demand $k$ is identified by a source node $s^k\in \mathcal N$, destination node $t^k\in \mathcal N$ and transmission rate $r^k$. Table~\ref{tab:notation_input} summarizes the notation used throughout the paper.
Our primary objective is to accept the maximum number of demands into the system; our secondary goal is to minimize the total delay, while each demand must fulfill all hard constraints on link and node capacity, jitter and packet loss probability. Moreover, at most one reflector can be used for each demand.
This translates into a Multi-Commodity Flow (MCF) problem, that can be expressed either via a link- or path-based formulation, according to the needs. We start off with the link-based formulation, that sets the decision variable $x^k_{i,j}=1$ if and only if the directed edge $(i,j)$ is used for routing demand $k$, i.e.,

{\small
\begin{align}
\min_{x} & \, \sum_{k\in\mathcal{K}} \sum_{(i,j)\in\mathcal{E}} x_{i,j}^k d_{i,j} + M \sum_{k\in\mathcal{K}} \sum_{(i,j)\in\overline{\mathcal{E}}} x_{i,j}^k \label{eq:objfun} \\
\mathrm{s.t.} & \, \sum_{j:(i,j)\in\mathcal E\cup \overline{\mathcal{E}}} x_{i,j}^k - \sum_{j:(j,i)\in\mathcal E\cup \overline{\mathcal{E}}} x_{j,i}^k = \gamma_i^k, \quad \forall\, i\in\mathcal{N}, k\in \mathcal K \label{eq:flow_conserv} \\
& \, \sum_{k\in \mathcal K} x_{i,j}^k r^k \le b_{i,j}, \quad \forall \, (i,j)\in \mathcal E \cup \overline{\mathcal{E}} \label{eq:capacity} \\
& \, \sum_{(i,j)\in\mathcal E \cup \overline{\mathcal{E}}} x_{i,j}^k z_{i,j} \le Z^k, \quad \forall\, k\in \mathcal K \label{eq:jitter} \\
& \, \sum_{k\in\mathcal K} \sum_{j:(i,j)\in\mathcal E \cup \overline{\mathcal{E}}} x_{i,j}^k r^k \le N_i, \quad \forall\, i\in\mathcal N \label{eq:processing1} \\
& \, \sum_{k\in\mathcal K} \sum_{j:(j,i)\in\mathcal E \cup \overline{\mathcal{E}}} x_{j,i}^k r^k \le N_i, \quad \forall\, i\in\mathcal N  \label{eq:processing2} \\
& \, \sum_{(i,j)\in\mathcal E \cup \overline{\mathcal{E}}} x_{i,j}^k \log f_{i,j} \ge \log F^k, \quad \forall\, k\in \mathcal K \label{eq:prob_succ} \\
& \, \sum_{(i,j)\in\mathcal E \cup \overline{\mathcal{E}}} x_{i,j}^k \le 2, \quad \forall\, k\in \mathcal K \label{eq:maxhop} \\
& \, x_{i,j}^k \in \{0,1\}, \quad \forall\,i,j,k. \label{eq:fin}
\end{align}
}
\normalsize
We remark that, via the formulation in (\ref{eq:objfun}-\ref{eq:fin}), we implicitly augmented the original graph $G$ into a clique, where all the new (artificial) edges $\overline{\mathcal E}$ have a large delay $M$. More specifically, $M$ should be set as larger than $2|\mathcal K|\max_{i,j}d_{i,j}$ in order to prioritize the maximization of the number of accepted demands over the delay minimization goal. 
Eq. \eqref{eq:flow_conserv} describes the standard flow conservation constraints, where $\gamma_i^k=1$ if $i=s^k$, $\gamma_i^k=-1$ if $i=t^k$ and $\gamma_i^k=0$ otherwise. Eqs. \eqref{eq:capacity} and \eqref{eq:jitter} account for the capacity and jitter hard constraints, respectively. $Z^k$ is the maximum jitter value for demand $k$.  Expressions \eqref{eq:processing1},\eqref{eq:processing2} ensure that each node $i$ processes traffic at a rate not exceeding $N_i$. Eq. \eqref{eq:prob_succ} claims that the probability of successful transmission for demand $k$ is at least $F^k$. Finally, Eq. \eqref{eq:maxhop} translates the requirement that at most one reflector is used for each demand, i.e., the maximum number of hops equals 2. 

\begin{proposition} \label{prop:np_har}
The overlay routing problem formalized in Eqs.\eqref{eq:objfun}-\eqref{eq:fin} is NP-Hard.
\end{proposition}

\begin{proof}
We prove that overlay routing problem is NP-hard by considering a simplified instance of the problem where QoS constraints~\eqref{eq:jitter}-\eqref{eq:maxhop} are neglected.
In other words, we do not limit the set of feasible paths to those that satisfy the QoS constraints. 
In this case, the overlay routing problem becomes a Multi-Commodity Integral Flow problem, which is known to be NP-hard~\cite{even1975complexity}. Thus, the overlay routing problem contains an NP-hard problem as special case, which makes the overlay routing problem itself NP-hard.
\end{proof}

We now present the path-based formulation for our optimization problem, that is equivalent to the link-based one in Eqs. \eqref{eq:objfun}-\eqref{eq:fin}. As explained in the next section, this formulation enables the decomposition of the original overlay routing problem into simpler and smaller sub-problems that can be solved more efficiently.
We define $\mathcal P^k$ as the set of paths from source $s^k$ to destination $t^k$ fulfilling the constraints on number of hops ($\le 2$), jitter ($\le Z^k$) and probability of successful transmission ($\ge F^k$). We call $\mathcal P_i^k\subseteq \mathcal P^k$ the set of paths visiting node $i\in\mathcal N$. Similarly, $\mathcal P_e^k\subseteq \mathcal P^k$ is the set of paths crossing edge $e\in\mathcal E$. Moreover, let $d_p$ be the total delay of path $p$, by accounting that edges not in $\mathcal E$ have delay $M$. Then, the path-based formulation for MCF writes as follows:

{\small
\begin{align}
\min_{y} & \, \sum_{k\in\mathcal K} \sum_{p\in\mathcal P^k} y_p d_p  \label{eq:path:objfun}\\
\mathrm{s.t.} & \, \sum_{k\in\mathcal K} \sum_{p\in\mathcal P^k_e} y_p r^k\le b_e , \quad \forall\,e\in\mathcal E \label{eq:path:capacity} \\
& \, \sum_{k\in\mathcal K} \sum_{p\in\mathcal P^k_i} y_p r^k\le N_i , \quad \forall\,i\in\mathcal N \label{eq:path:processing} \\
& \, \sum_{p\in\mathcal P^k} y_p = 1, \quad \forall\, k\in\mathcal K \label{eq:path:single_path} \\
& \, y_p \in\{0,1\}
\end{align}
}
\normalsize
where $y_p=1$ whenever path $p\in\mathcal P^k$ is used for $k\in\mathcal K$.

We observe that the two set of constraints~\eqref{eq:path:capacity}-\eqref{eq:path:processing} represent the transmission and processing capacity limits of links and nodes, respectively. 
%


\section{Path finding algorithm for fast transfers }
\label{algorithms}
As illustrated in the previous section the overlay routing problem is NP-Hard. Therefore, the computational time steeply increases with the network size and the number of demands.
To solve the overlay routing problem even in large scale scenarios, we propose to decompose the original problem into simpler subproblems that can be solved more efficiently.
More specifically, we first relax the integrality constraints for the variable $y_{p}$, enabling the use of multiple paths for the routing of each demand, and then we design a column generation algorithm~\cite{NET:NET3230140406} to solve the underlying MCF problem.
In order to deal with the QoS constraints, we use a pseudo-polynomial algorithm to generate only shortest paths that satisfy constraints~\eqref{eq:jitter},~\eqref{eq:prob_succ}, and~\eqref{eq:maxhop}.
Finally, we design a randomized method to assign a single path to those demands whose LP solution consists in splitting the traffic over multiple paths.

A general description of the steps of the proposed overlay routing method, referred to as \emph{Column Generation-Generalized LARAC} (CG-GLC), is illustrated in Alg.~\ref{alg:cg_alg}. In the next subsections we provide more details. 

\begin{algorithm}[htb]
	\small
	\caption{CG-GLC}
	\label{alg:cg_alg}
	\KwIn{$d_{e}$, $b_{e}$,$G\left(\mathcal{N}, \mathcal{E}, w(\cdot) \right)$}
	\KwOut{$\boldsymbol{y}$}
	\tcc{Routing}
	\nl Find an initial feasible solution $y_{p}$ for the reduced master problem~\eqref{eq:path:objfun}-\eqref{eq:path:processing}\label{alg:cg_alg:init}\;
	\nl Compute the dual point $\boldsymbol{\mu} = \left[ \boldsymbol{\lambda}, \boldsymbol{\sigma} \right]$ corresponding to $\boldsymbol{y}$\;
	\nl \While { $\boldsymbol{\mu}$ is not feasible }{
		\For { $k \in \mathcal{K}$ }{
			Generate a graph $G\left(\mathcal{N}, \mathcal{E}, w(\cdot) \right)$ with links weights equal to $w_{e} = d_{e} + r_{k} \lambda_{e}$\;
			Compute the constrained shortest path $p$ over $G$\label{alg:cg_alg:csp}\;
			\If { $\sigma_{k} \leq \suml_{e \in p} w_{e}$  }{ \tcc{This var. improves \eqref{eq:path:objfun}}
				Add $y_{p}$ to the reduced master problem~\eqref{eq:path:objfun}-\eqref{eq:path:processing} (primal problem)\;
			}
		}
		Solve the new reduced master problem and get the new solution $\boldsymbol{y}$\;
		Compute the dual point $\boldsymbol{\mu}$ corresponding to $\boldsymbol{y}$\label{alg:cg_alg:end}\;
	}
	\tcc{Rounding}
	\nl\label{alg:cg_alg:rounding} \While { solution $\boldsymbol{y}$ has changed }{
		\For { $k \in \mathcal{K}$, $\exists p \in \mathcal P_{k} : 0 < y_{p} < 1$ }{
			Select path $p$ with probability $y_{p}$\;
			\If {any of the capacity constraints \eqref{eq:path:capacity} is violated }{
				Restore the value of $y_{p}$\hspace{2ex}$\forall p \in \mathcal P_{k}$\;
			} \Else {
			Set $y_{p} = 1$\;
			Set $y_{q} = 0$\hspace{2ex}$\forall q \neq p$\;
			Reduce the capacity of link in the path $p$ by $r_{k}$\;
		}
	}
}
return $\boldsymbol{y}$\;
\end{algorithm}
\normalsize

\subsection{Solving the Constrained MCF problem}
The column generation technique enables us to consider only a subset of decision variables in the primal formulation at each iteration, and it uses the dual formulation to include only those variables that can improve the objective function.
Indeed, from the duality theory we know that every feasible point to the dual problem $\boldsymbol{\mu}^{*}$ gives a lower bound on the optimum value of the primal $\boldsymbol{y}^{*}$, and every feasible point to the primal problem $\boldsymbol{y}^{*}$ gives an upper bound on the optimal value of the dual $\boldsymbol{\mu}^{*}$.
Therefore, a feasible primal solution $\boldsymbol{y}^{*}$ is optimal if the corresponding dual point $\boldsymbol{\mu}^{*}$ is feasible. 
Our algorithm exploits this property in order to consider, for each demand, only a small set of variables representing feasible paths and add new variables to the primal formulation as long as the corresponding dual solutions are unfeasible.
In our problem, the vector of dual variables $\boldsymbol{\mu} = \left[ \boldsymbol{\lambda}, \boldsymbol{\sigma} \right]$ is split into two different sets of variables, where $\boldsymbol{\lambda}$ corresponds to the capacity constraints~\eqref{eq:path:capacity}-\eqref{eq:path:processing} and $\boldsymbol{\sigma}$ corresponds to constraint~\eqref{eq:path:single_path}, which indicates that a subset of $\mathcal{P}^k$ is used to route a demand.

We underline that, as long as the paths generated during the column generation procedure satisfy all QoS constraints, the final solution computed by our algorithm is optimal for the LP relaxation of the overlay routing problem, in the sense that each demand is fully satisfied by one or multiple paths.
To this aim, we use the GEN-LARAC (GLC) algorithm as a subroutine for solving the constrained shortest path problem~\cite{XiThXu05}, since it computes in pseudo-polynomial time a path that satisfies all QoS constraints~\eqref{eq:jitter},~\eqref{eq:prob_succ}, and~\eqref{eq:maxhop}.
Therefore, at the end of the column generation algorithm we only have  to choose a single path for each demand, without reconsidering all constraints of the original problem.

If $\boldsymbol{\lambda} \geq 0$, the constraints of the dual problem can be formulated as follows:
\begin{equation}
	\sigma_{k}^{*} - \suml_{e \in p} r_{k} \lambda_{e}^{*} \leq d_{p},	\quad  \forall k \in \mathcal{K}, p \in \mathcal P_{k},
\label{eq:path:dual_constr}
\end{equation}
Eq.~\eqref{eq:path:dual_constr} states that in a feasible point of the dual problem $\boldsymbol{\mu}^{*} = \left[ \boldsymbol{\lambda}^{*}, \boldsymbol{\sigma}^{*} \right]$, there exists no path for any demand such that 
$\sigma_{k}^{*} > \suml_{e \in p} r_{k} \lambda_{e}^{*} + d_{p}$, which can be rewritten as $\sigma_{k}^{*} >  \suml_{e \in p}\left( r_{k}  \lambda_{e}^{*} + d_{e}\right)$. 
In other words, there exists no path for any demand that can further improve the objective function~\eqref{eq:path:objfun}.
In contrast, if such a path exists, i.e., $\exists k \in \mathcal{K}, p \in \mathcal P_{k}$:~$\sigma_{k}^{*} > r_{k}  \suml_{e \in p}\left( r_{k}  \lambda_{e}^{*} + d_{e}\right)$, 
then the dual point $\boldsymbol{\mu}^{*}$ is unfeasible and the objective function of the primal problem can be further reduced by considering such a path (recall that the primal solution is an upper bound of the dual solution).

All we need to do at each iteration of the column generation algorithm is checking whether the condition 
$\sigma_{k}^{*} > r_{k}  \suml_{e \in p} \left( r_{k}  \lambda_{e}^{*} + d_{e}\right)$ holds for all possible paths of all demands (i.e., $\forall \, p \in \mathcal P_{k}$).
We observe that the computation of $\suml_{e \in p} \left( r_{k}  \lambda_{e}^{*} + d_{e}\right)$ can be performed efficiently using an algorithm for the constrained shortest path computation on the weighted graph 
$G\left(\mathcal{N}, \mathcal{E}, w(\cdot) \right)$, where the link weight function $w(\cdot):\mathcal{E}\rightarrow \mathbb{R}_{\geq 0}$ is computed as $w_{e} = r_{k} d_{e} + \lambda_{e}$.

We point out that, in order to maintain the feasibility of this routine, we also add some \emph{dummy paths} on which demands, rejected by column generation, can be allocated. 

\subsection{Rounding procedure}

To address the splitting of demands over multiple paths, which can be caused by the resolution of the LP relaxation of the problem \eqref{eq:path:objfun}-\eqref{eq:path:processing}, we introduce a randomized rounding phase at the end of the column generation algorithm, in order to convert the fractional solution into a feasible integer solution.
Whenever multiple paths are used to route a demand from its origin to its destination, a single route is selected with probability equal to the portion of the demand allocated to each specific path.
In particular, for each demand $k$ that has been split, we consider all possible paths where $y_{p} > 0$ and we select a path $p$ according to the overall flow allocated to it.
If the whole demand $k$ can be transmitted at its nominal rate $r_{k}$ over the selected path $p$ without violating any capacity constraint, we keep only the path $p$ by fixing the variable $y_{p} = 1$ and all other variables corresponding to alternative paths to $0$. Finally, we reduce the capacity of the links that belong to $p$ by the demand's rate $r_{k}$.
The randomized routing procedure terminates either when a single path has been allocated to all demands or when the solution $\boldsymbol{y}$ does not change between two consecutive iterations. In this latter case, all demands that are still split over multiple paths are rejected.


\subsection{Multicast Overlay routing}

Our CDN framework can be adapted to solve the multicast overlay routing problem, typical of scenarios where overlay nodes can treat video streams at application layer by decapsulating, processing, and re-encapsulating them before forwarding. In this case, the operator can save  bandwidth in the overlay network by building a multicast tree among the source and the nodes that are interested in the same content.
A demand $k$ is therefore identified by a source node $s^k\in \mathcal N$, multiple destination nodes $\mathcal T^k \subset \mathcal N$ and a transmission rate $r^k$. 
The goal is to connect a single source $s^k$ to several destinations $t^k \in \mathcal T^k$ using the multicast tree that satisfies all QoS constraints and has the minimum delay.
To this end, we redefine $\mathcal P^k$ as the set of trees connecting source $s^k$ to its destinations $t^k \in \mathcal T^k$ fulfilling all QoS constraints, while binary variable $y_p$ indicates whether the tree $p \in \mathcal P^k$ has been selected.
The main difference with respect to the previous overlay routing problem dwells in the subroutine used to generate constrained trees with minimum delay. In this case, this problem becomes an extension of the Steiner tree problem, which is known to be APX-complete. Therefore, we cannot compute a multicast tree arbitrarily close to optimum in polynomial time, but we can only use schemes that compute a solution in polynomial time within an constant approximation factor like~\cite{byrka2010improved}.
We observe that even if the algorithm will not converge to the optimal fractional solution of the multicast overlay routing problem, it can still improve the starting solution computed by a heuristic approach like the one proposed in \cite{zhou2015joint}.

\section{Performance Results}
\label{results}

In this section, we analyze the performance our Column Generation (CG) approach based on GLC (CG-GLC). 
In order to show the distance of CG-GLC from the optimal solution, also called ``optimality gap'', we use the solution of the ILP presented in Section \ref{secIIIB_formulation} as a benchmark.

Results are evaluated in two different scenarios: (i) \textit{Telco CDN} where we used traces collected on a real overlay network of a Telco CDN and (ii) \textit{Synthetic CDN} where we generated a larger random network. We present a performance evaluation over time considering as indicators the percentage of accepted demands, the running time and the average delay. 
We point out that in both scenarios, nodes can be either source/destination of a demand or reflector if the demand is not originated in or destined to the considered node. Without loss of generality, for each demand, we considered $F^k> 99\%$ and 	$Z^k< 2$ s.

All the tests have been executed on a machine with Intel Core i7-5600U \@ 2.6 GHz with 8 GB of RAM. The C++ libraries used to build and update the network graph have been provided by Lemon. The ILP and the resolution of the reduced master problem have been carried out using CPLEX.

\subsection{Real CDN Overlay Network} \label{sec:cdn_overlay_results}

In this subsection, first we present the traces used for simulations and then we compare the performance of CG-GLC and ILP via numerical evaluations.
In Table~\ref{tab:network_parameters_CDN}, we list the main parameters considered in the CDN overlay scenario. 
\begin{table}[t!]
	\centering
	\small
	\begin{tabular}{|c|c|l|}	
		\hline
		$\mathcal{N}$ & 11 & Number of nodes\\
		\hline	
		$\mathcal{E}$  & 82 & Number of links \\
		\hline
		$b_{i,j}$ & 50  & Link capacity [Gbps]\\
		\hline
		$N_i$ & 150  & Node processing capacity [Mbps]\\
		\hline
		$\mathcal{K}$ & 82 & Number of demands \\
		\hline
		$r_k$ & 50 & Transmission rate [Mbps]\\
		\hline
		$F^k$ & 99\% & Minimum prob. of successful transmission \\
		\hline
		$Z^k$ & 2 & Maximum jitter for demand $k$ [s]\\
		\hline
	\end{tabular}
	\caption{Parameters for the Telco CDN scenario.}
	\label{tab:network_parameters_CDN}
\end{table} 

Transmission rate and node processing capacity are set to 50 Mbps and 150 Mbps, respectively, which leads to a scenario where nodes are fairly stressed.
The link QoS metrics used for our experiments, as delay, jitter and packet loss, are described in Subsection~\ref{sec:dataset}.	

\subsubsection{Dataset presentation} \label{sec:dataset}

We used real traces from a Telco CDN operator. Since the QoS of the considered overlay network varies over time, a monitoring system carries out periodic probing between overlay nodes. In this subsection, we show the results of QoS metrics monitoring considering an aggregation period of 5 minutes.

\begin{figure}[!ht]
	\centering
	\subfloat[][Distribution of pairwise average delay with bins of 5 ms.]{\includegraphics[width=0.48\textwidth]{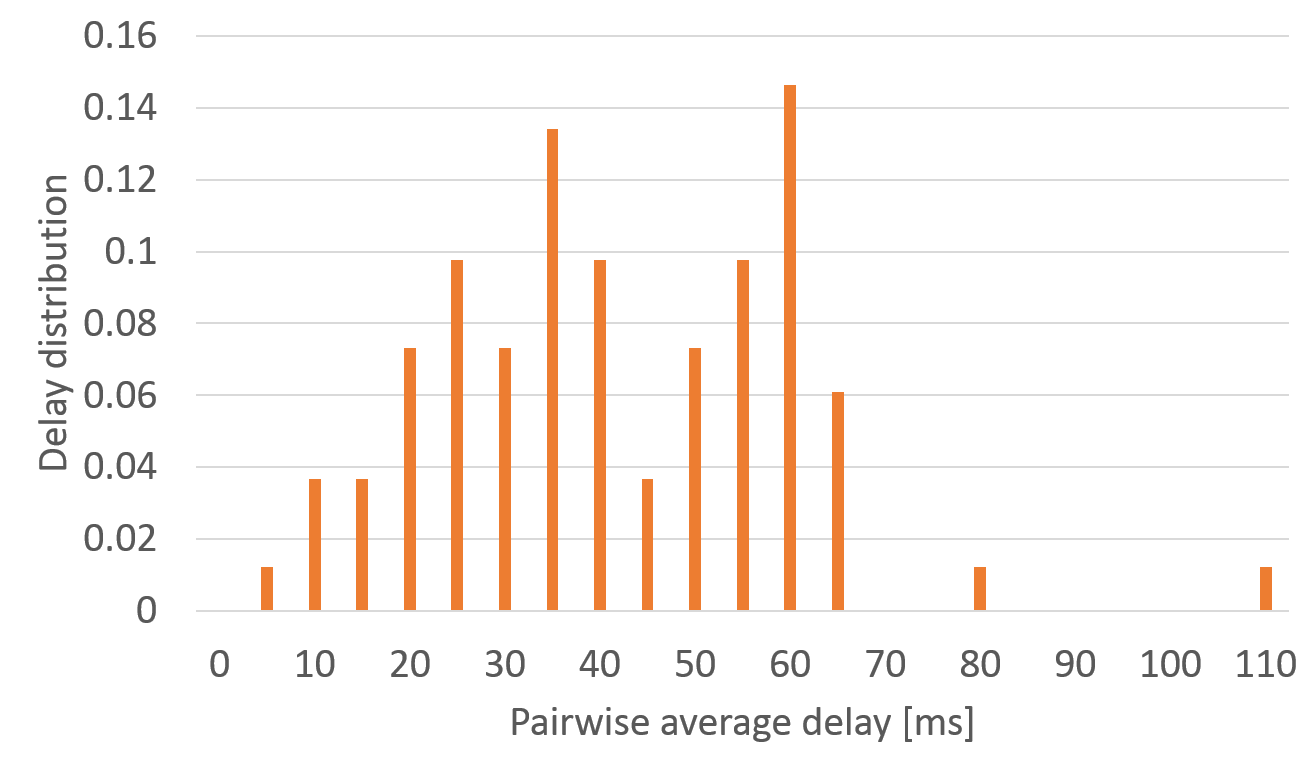}\label{fig:delay_distribution} } \\	
	\subfloat[][Distribution of pairwise average jitter with bins of 5 ms.]{\includegraphics[width=0.48\textwidth]{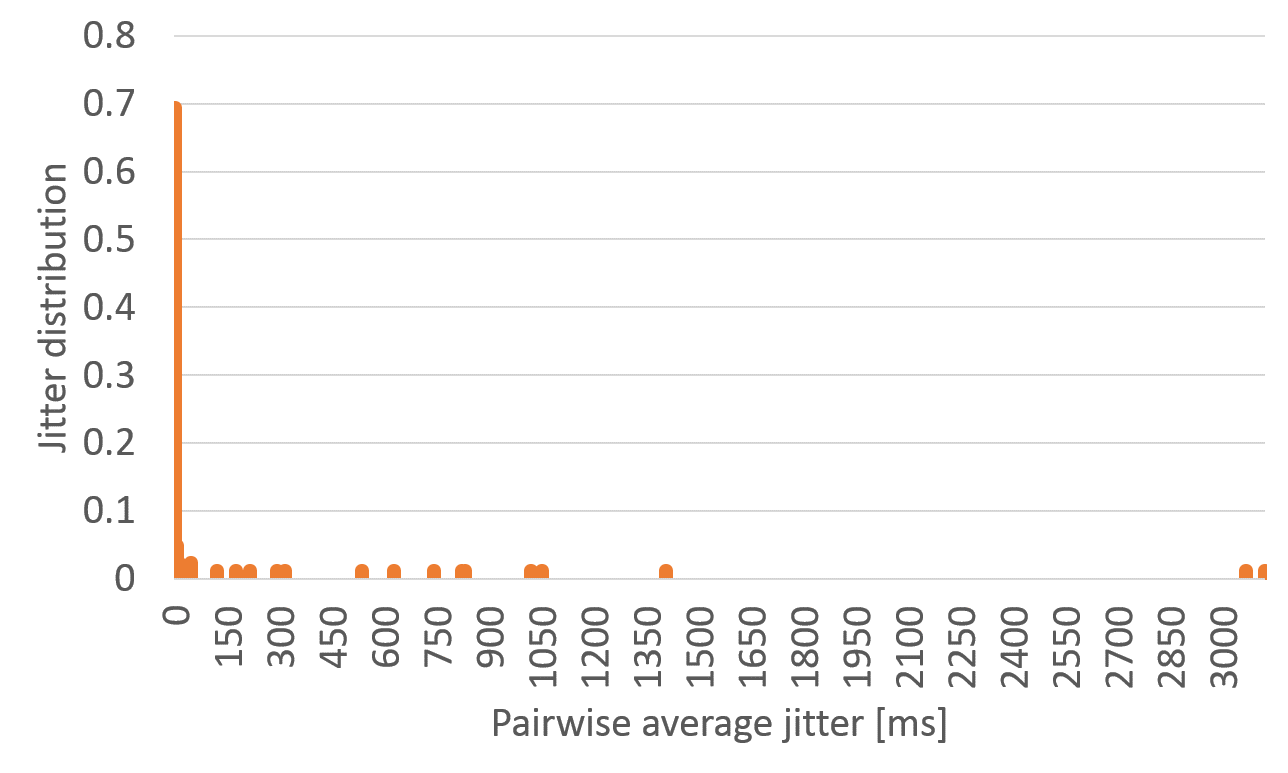}\label{fig:jitter_distribution}}  \\
	\subfloat[][Distribution of pairwise average packet loss with bins of  0.1 \%.]{\includegraphics[width=0.48\textwidth]{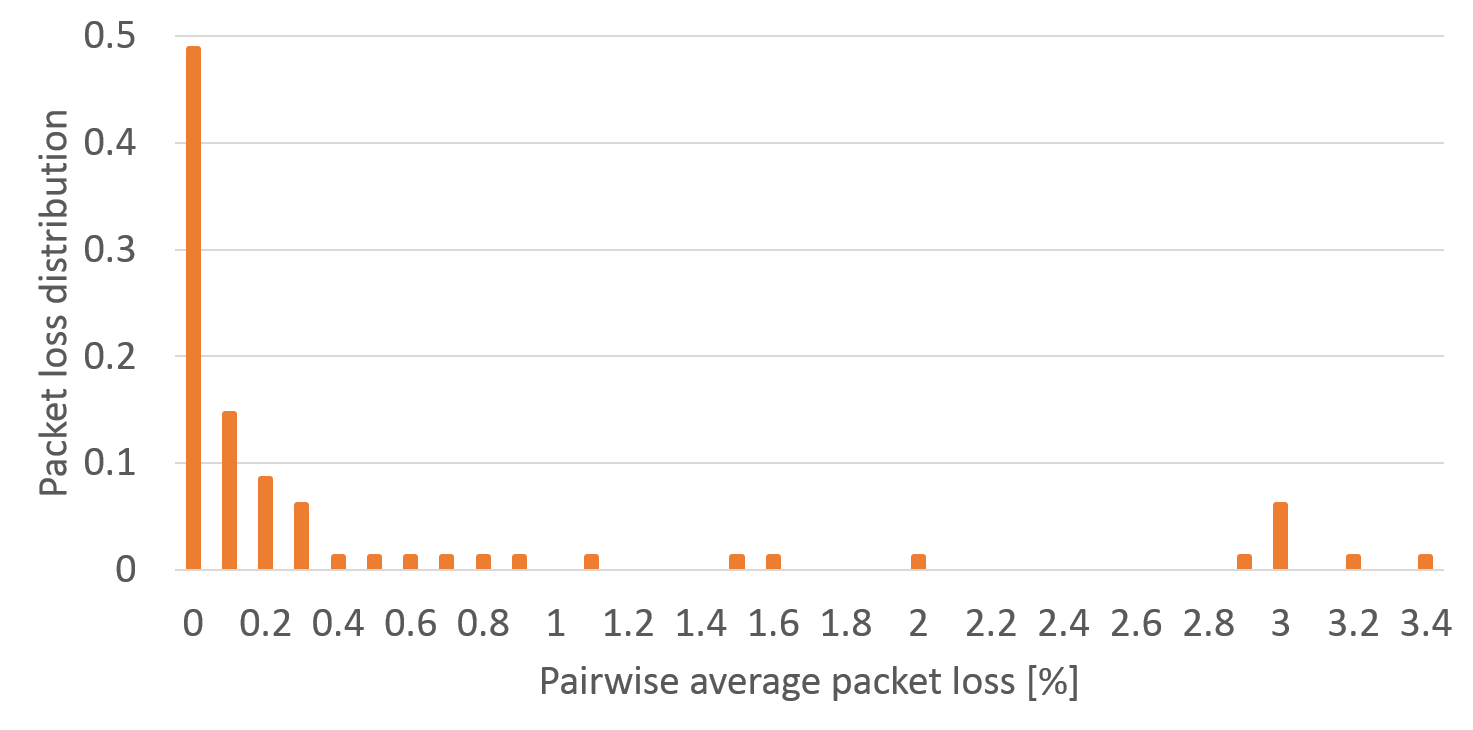}\label{fig:ploss_distribution}}
	\caption{Overlay link statistics of the Telco CDN overlay.}
	\label{fig:traces_histogram}
\end{figure}

In Fig.~\ref{fig:delay_distribution}, we show the pairwise distribution of the average delay. We can see that the support of delay distribution is 5-80 ms, while only few links present larger delay. The jitter distribution presented in Fig.~\ref{fig:jitter_distribution} shows that the jitter is massively concentrated in the first bins. In fact, 90\% of the jitter samples are smaller than 330 ms.
Finally, as shown in Fig.~\ref{fig:ploss_distribution}, most of the links have an average packet loss smaller than 0.3\%, although 6\% of the links have a packet loss of 3\%. This is due to local perturbations that caused consistent packet losses.

\subsubsection{Performance Evaluation}

We now compare the performance of CG-GLC and the solution of ILP in two scenarios. The former is the one described before, where an overlay network is considered (\emph{overlay} in the Figures). The latter neglects the presence of the overlay, while only relying on the direct path between source and destination (\emph{direct}).

We compare CG-GLC and ILP solutions in terms of percentage of accepted demands and run-time in Fig.~\ref{fig:accDMD_cdn} and Fig.~\ref{fig:runtime_cdn}, respectively. At each time instant we try to allocate all the 82 demands. As the network capacity is not saturated, the node processing capacity is the real bottleneck of the system (i.e., constraints~\eqref{eq:processing1} and~\eqref{eq:processing2}). 

\begin{figure}[!ht]
	\centering
	\includegraphics[width=0.48\textwidth]{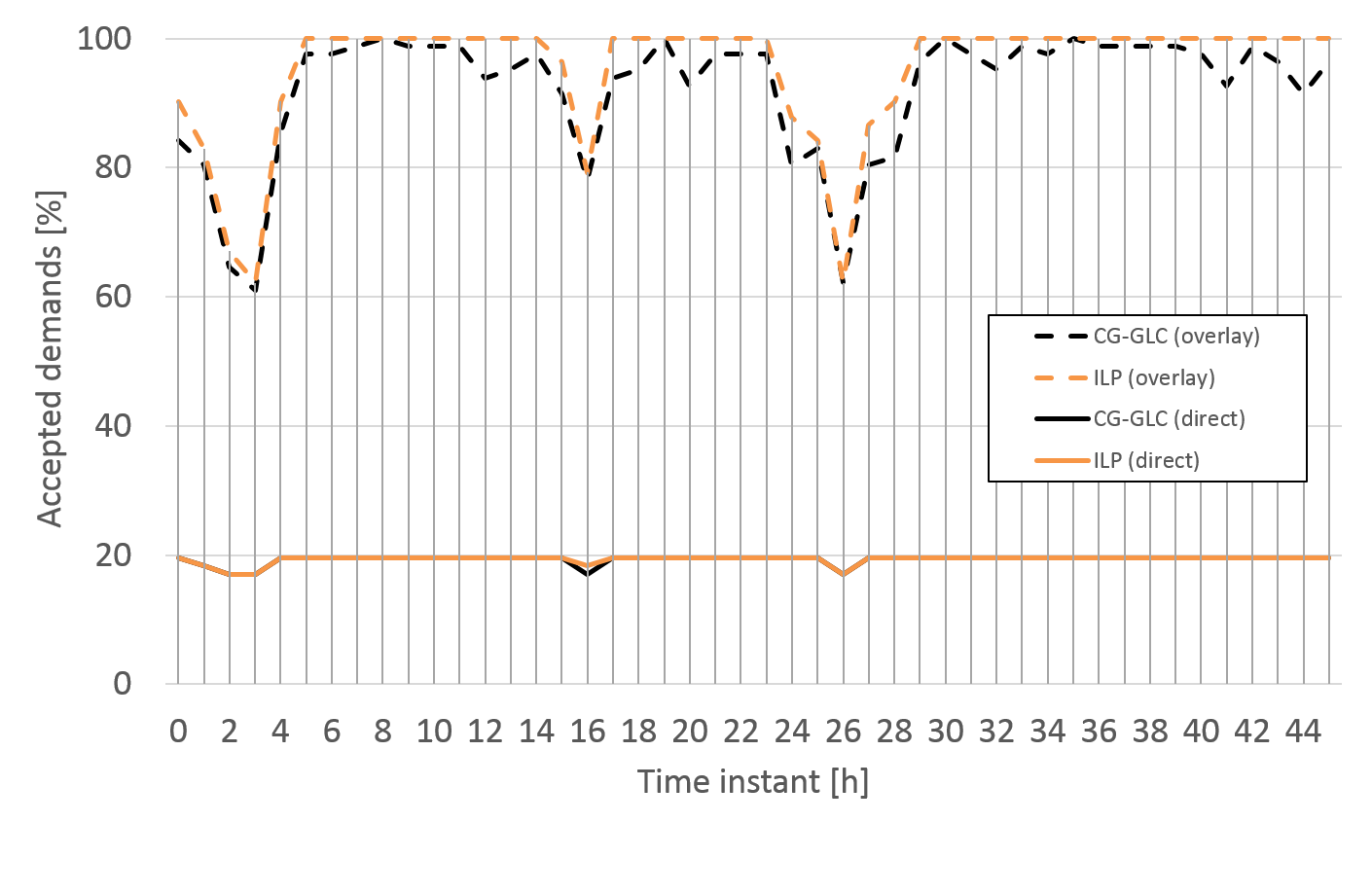}
	\caption{Percentage of accepted demands over time for the Telco CDN scenario.}
	\label{fig:accDMD_cdn}
\end{figure}

We first remark that while through the overlay (\emph{overlay}) we manage to accept in average 90\% of the demands, without the overlay (\emph{direct}) we cannot accept more than 20 \% of the demands. Accepted demands are flows for which QoS constraints in term of jitter and packet loss are met. This result highlights the benefit of overlay networking compared to the case when one relies only on the underlay. Indeed, it increases path diversity and the chances to find feasible paths. 

Secondly, we observe that our CG-GLC solution strikes a good performance/complexity trade-off. In fact, CG-GLC is characterized by an average optimality gap of 3\% in terms of  percentage of accepted demands. Moreover, CG-GLC has a running time around 10 times smaller than ILP solution, that was produced by standard commercial software CPLEX. This is due to the fact that CG-GLC, instead of solving the whole problem, focuses only on a subset of the original problem (faster to be solved), by adding to it only the solutions (i.e., paths) which improve the objective function. 
We point out that the CG-GLC can compute an entire network reconfiguration in less than 200 ms.

As the number of accepted demands is not the same for CG-GLC and ILP, it is not possible to carry out a fair comparison between the two approaches in terms of average delay.
As shown in Fig.~\ref{fig:avgDelay_cdn}, the average delay without overlay (\emph{direct}) is smaller than the one with the overlay (\emph{overlay}).  This is because the number of accepted demands is smaller without the overlay and, according to the CG routine, CG-GLC tries to use first paths with low delay but, as long as new demands are accepted, they are routed on more ``expensive'' paths.


\begin{figure}[!ht]
	\centering
	\includegraphics[width=0.48\textwidth]{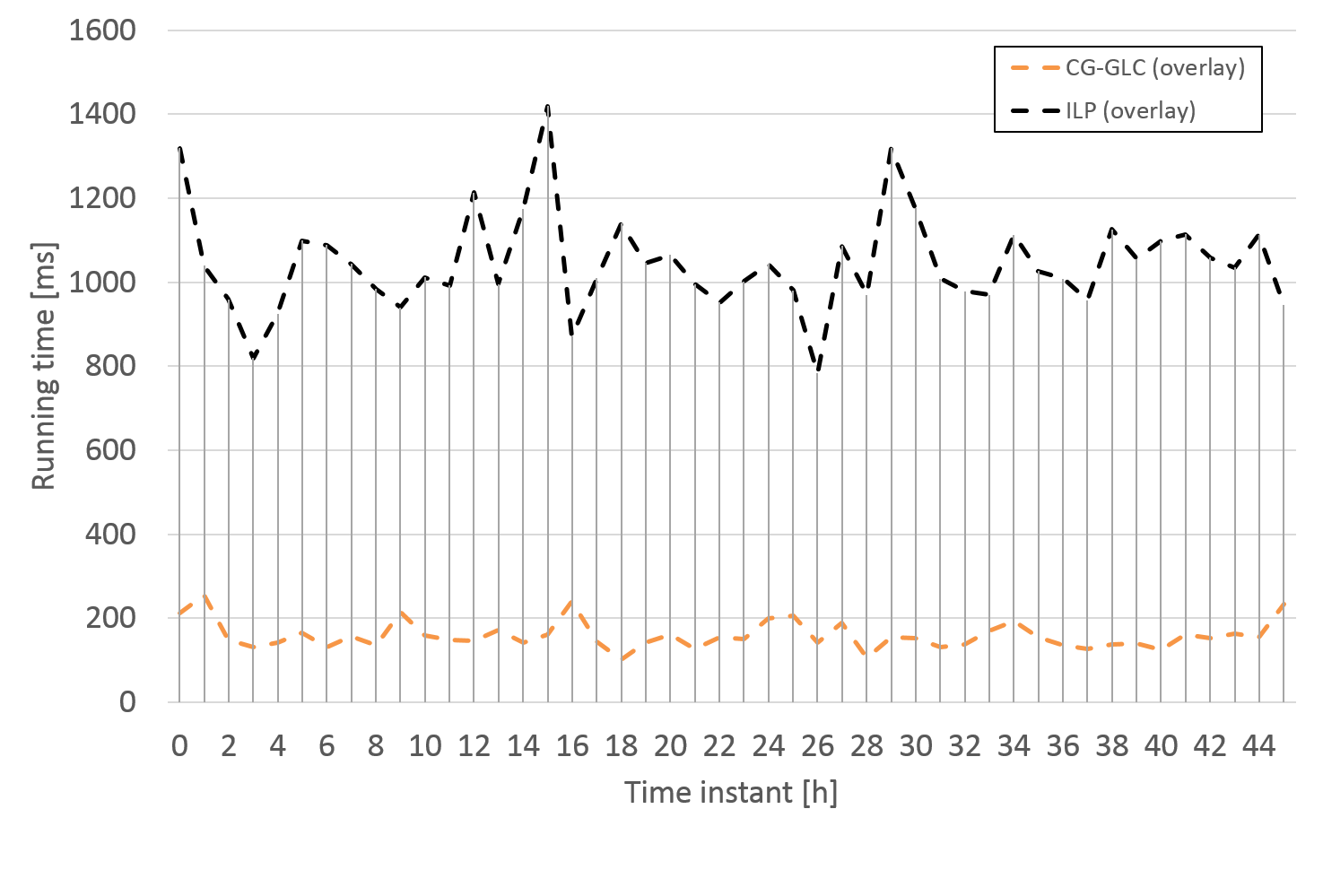}
	\caption{Running time of the CG-GLC and ILP over time for the Telco CDN scenario.}
	\label{fig:runtime_cdn}
\end{figure}



\begin{figure}[!ht]
	\centering
	\includegraphics[width=0.42\textwidth]{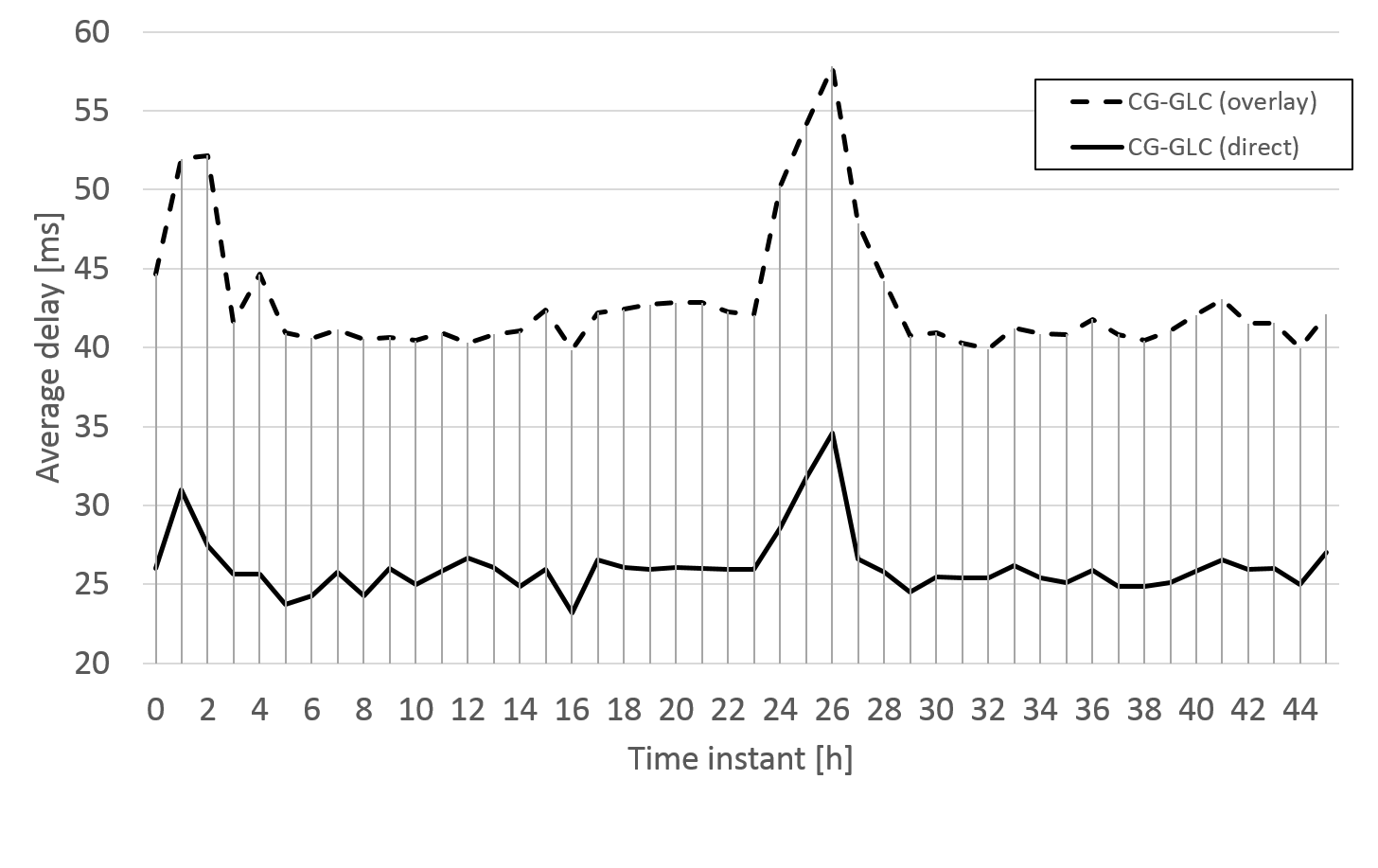}
	\caption{Average delay over time for CG-GLC for the Telco CDN scenario.}
	\label{fig:avgDelay_cdn}
\end{figure}


\subsection{Synthetic Overlay Network}
\label{sec:synthetic}

In this section, we further evaluate the performance of the proposed CG-GLC algorithm, in terms of percentage of accepted demands and running time, on a synthetic CDN network generated with a Barabasi model. In this scenario as well, we compare the results with (\emph{overlay}) and without (\emph{direct}) overlay network.
We point out that the ILP results are only provided for less than 150 demands, due to simulation time constraints: the ILP solution running time quickly explodes for bigger instances.
In Table~\ref{tab:network_parameters_barabasi}, we show the parameters of this experiment. 
For each demand, source and destination nodes are randomly chosen among the set of nodes. Transmission rate of demands are distributed as an exponential random variable of parameter $50$~Mbps.

\begin{table}[t!]
	\centering
{\small
	\begin{tabular}{|c|c|l|}	
		\hline
		$\mathcal{N}$ & 50 & Number of nodes\\
		\hline	
		$\mathcal{E}$  & 450 & Number of links \\
		\hline
		$b_{i,j}$ & $\sim Unif\,[300,500]$ & Link capacity [Mbps]\\
		\hline
		$d_{i,j}$ & $\sim Unif\,[1,500]$ & Link delay [ms]\\
		\hline
		$\alpha_{i,j}$ &$\sim Unif\, [0,50] $& Link jitter [ms]\\
		\hline
		$f_{i,j}$ &$\sim Unif\,[0,0.2]$ & Link packet loss\\
		\hline
		$N_i$ & 400 & Node processing capacity \\
		&&[Mbps] \\
		\hline
		$\mathcal{K}$ & [90, 110, ... , 210] & Number of demands \\
		\hline
		$r_k$ & $\sim exp(50)$ & Transmission rate [Mbps]\\
		\hline
		$Z^k$&$\sim Unif\, [25,200]$& Max jitter [ms]\\
		\hline
		$F^k$&$\sim Unif\, [0.1,0.8]$& Max packet loss\\
		\hline
	\end{tabular}
}
	\caption{Network parameters for the Synthetic CDN scenario.}
	\label{tab:network_parameters_barabasi}
\end{table}

In Fig.~\ref{fig:accDMD_barabasi}, the percentage of accepted demands is presented as a function of the number of demands. 
\begin{figure}[!ht]
	\centering
	\includegraphics[width=0.42\textwidth]{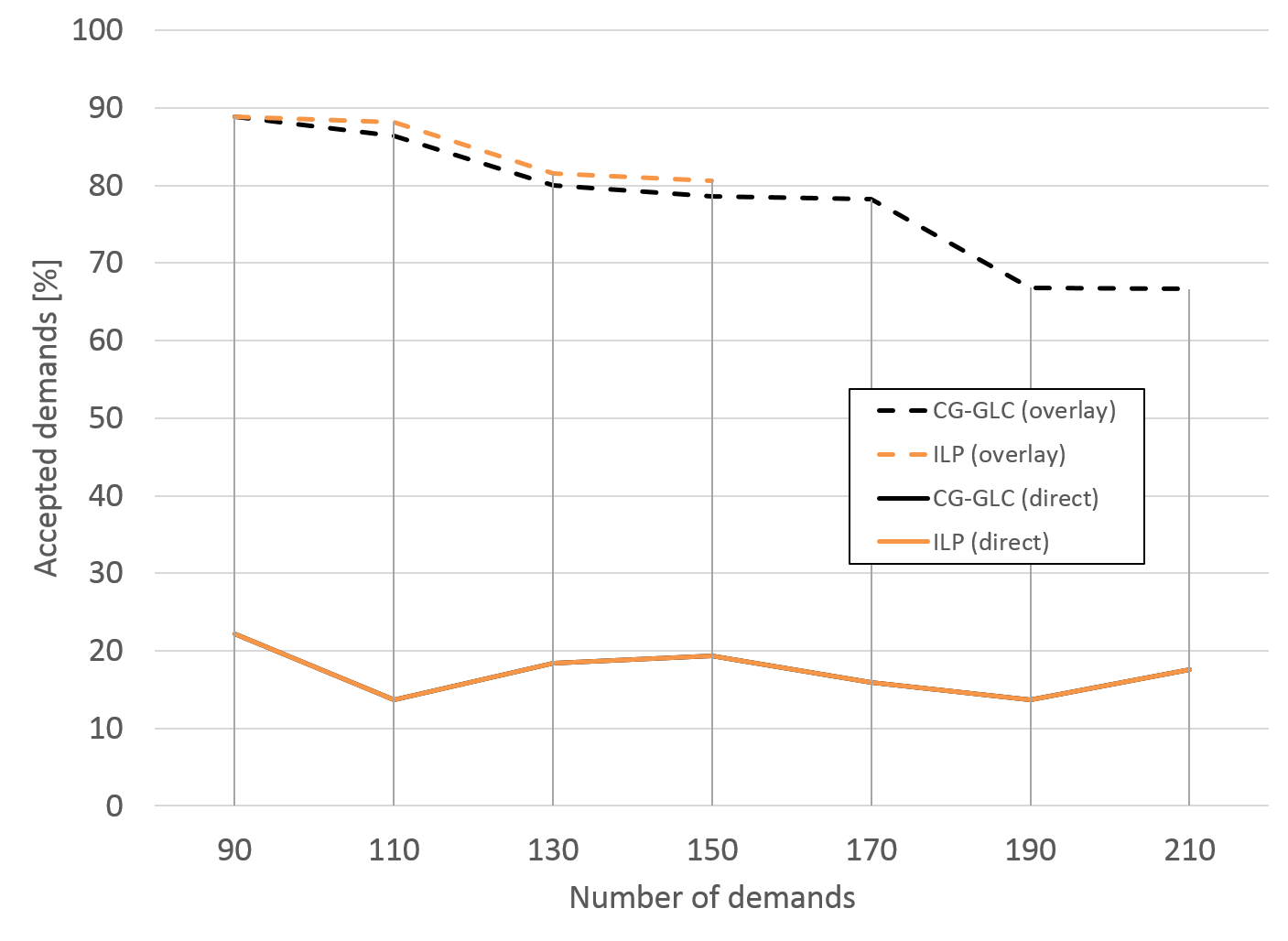}
	\caption{Percentage of accepted demands as a function of the number of demands for the Synthetic CDN.}
	\label{fig:accDMD_barabasi}
\end{figure}
Both with and without overlay network, the two algorithms have the similar performance as they manage to allocate the same number of demands, meaning that CG-GLC is very close (less than 2 \%) to the optimum provided by the ILP. However, when the load increases (i.e., $\mathcal{K}>150$), in the scenario with overlay network the ILP is no longer able to provide solutions in reasonable time, while CG-GLC confirms to be a valid approach. In the case without overlay, instead, many demands are rejected because of node processing capacity or absence of direct links between source and destination. For such a reason, CPLEX is able to solve the ILP even for larger sets of demands.

In Fig.~\ref{fig:runtime_barabasi} we compare the two algorithms in terms of running time for the network with overlay.
\begin{figure}[!ht]
	\centering
	\includegraphics[width=0.42\textwidth]{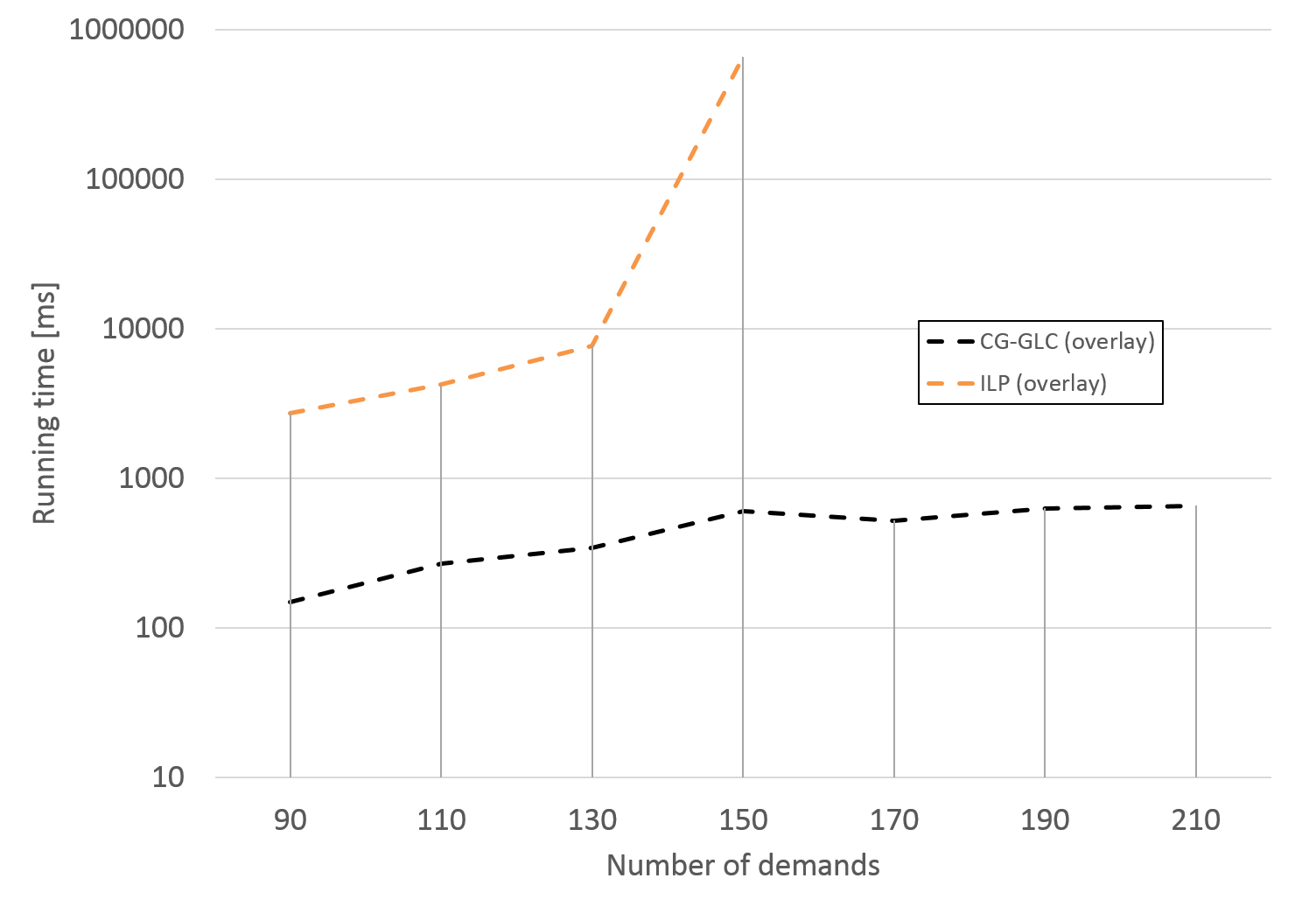}
	\caption{Running time of the two algorithms as a function of the number of demands for the Synthetic CDN scenario.}
	\label{fig:runtime_barabasi}
\end{figure}
Also in this case, CG-GLC is faster than the ILP because the column generation routine is designed to  solve a smaller problem than the one solved by the ILP. In particular, as the running time of the ILP grows exponentially with the number of demands, we can rely on the ILP solver only for scenarios with less than 150 demands. The running time of the CG-GLC grows exponentially as well, but more slowly, making this approach valid for larger networks.

\section{Concluding remarks} \label{conclusion}

In this paper, we have addressed the problem of fast video delivery in CDN. We propose an algorithm to nearly optimally allocate and maintain paths in the overlay network. 
More specifically, we formulate the problem as a multi-commodity flow under several QoS constraints derived from the requirements of CDN overlays networks. 
The proposed solution applies to the fast video delivery problem for both personal live streaming and Video-on-Demand use cases. Our approach, based on column generation and randomized rounding, has been tested against the optimal solution computed by solving the associated ILP formulation with commercial software. 
We used real traces from a Telco CDN and a random network.
Results show that our approach is an excellent compromise in terms of running time and optimality.

\bibliographystyle{IEEEtran}

\begin{thebibliography}{10}
\providecommand{\url}[1]{#1}
\csname url@samestyle\endcsname
\providecommand{\newblock}{\relax}
\providecommand{\bibinfo}[2]{#2}
\providecommand{\BIBentrySTDinterwordspacing}{\spaceskip=0pt\relax}
\providecommand{\BIBentryALTinterwordstretchfactor}{4}
\providecommand{\BIBentryALTinterwordspacing}{\spaceskip=\fontdimen2\font plus
\BIBentryALTinterwordstretchfactor\fontdimen3\font minus
  \fontdimen4\font\relax}
\providecommand{\BIBforeignlanguage}[2]{{%
\expandafter\ifx\csname l@#1\endcsname\relax
\typeout{** WARNING: IEEEtran.bst: No hyphenation pattern has been}%
\typeout{** loaded for the language `#1'. Using the pattern for}%
\typeout{** the default language instead.}%
\else
\language=\csname l@#1\endcsname
\fi
#2}}
\providecommand{\BIBdecl}{\relax}
\BIBdecl

\bibitem{intro_Cisco_VNI}
Cisco, ``{Cisco Visual Networking Index: Forecast and Methodology,
  2014–2019},'' May 2015.

\bibitem{maggs2015algorithmic}
B.~M. Maggs and R.~K. Sitaraman, ``Algorithmic nuggets in content delivery,''
  \emph{ACM SIGCOMM CCR}, vol.~45, no.~3, pp. 52--66, 2015.

\bibitem{Lumezanu:2009}
C.~Lumezanu, R.~Baden, N.~Spring, and B.~Bhattacharjee, ``Triangle inequality
  variations in the internet,'' in \emph{{Proc. ACM IMC}}, 2009.

\bibitem{nygren2010akamai}
E.~Nygren, R.~K. Sitaraman, and J.~Sun, ``The akamai network: a platform for
  high-performance internet applications,'' \emph{ACM SIGOPS Operating Systems
  Review}, vol.~44, no.~3, pp. 2--19, 2010.

\bibitem{chiu2015we}
Y.-C. Chiu, B.~Schlinker, A.~B. Radhakrishnan, E.~Katz-Bassett, and
  R.~Govindan, ``Are we one hop away from a better internet?'' in \emph{{Proc.
  ACM IMC}}, 2015.

\bibitem{andersen2002resilient}
D.~Andersen, H.~Balakrishnan, F.~Kaashoek, and R.~Morris, ``Resilient overlay
  networks,'' \emph{ACM SIGCOMM CCR}, vol.~32, no.~1, 2002.

\bibitem{Nunes14}
B.~A.~A. Nunes, M.~Mendonca, X.~N. Nguyen, K.~Obraczka, and T.~Turletti, ``A
  survey of software-defined networking: Past, present, and future of
  programmable networks,'' \emph{IEEE Com. Surveys Tutorials}, 2014.

\bibitem{liu2012case}
X.~Liu, F.~Dobrian, H.~Milner, J.~Jiang, V.~Sekar, I.~Stoica, and H.~Zhang, ``A
  case for a coordinated internet video control plane,'' in \emph{{Proc. ACM
  SIGCOMM}}, 2012.

\bibitem{quic}
J.~Iyengar, I.~Swett, R.~Hamilton, and A.~Wilk, ``{QUIC: A UDP-Based Secure and
  Reliable Transport for HTTP/2},'' Internet-Draft
  draft-tsvwg-quic-protocol-02, Jan. 2016.

\bibitem{seufert2015survey}
M.~Seufert, S.~Egger, M.~Slanina, T.~Zinner, T.~Ho{\ss}feld, and P.~Tran-Gia,
  ``A survey on quality of experience of http adaptive streaming,'' \emph{IEEE
  Communications Surveys \& Tutorials}, vol.~17, no.~1, pp. 469--492, 2015.

\bibitem{begen2005multi}
A.~C. Begen, Y.~Altunbasak, O.~Ergun, and M.~H. Ammar, ``Multi-path selection
  for multiple description video streaming over overlay networks,'' \emph{Sig.
  Proc.: Image Com.}, vol.~20, no.~1, pp. 39--60, 2005.

\bibitem{thi2015qos}
T.~M. Thi, T.~Huynh, and W.-J. Hwang, ``Qos-enabled streaming of multiple
  description coded video over openflow-based networks,'' \emph{Nonlinear
  Theory and Its Applications, IEICE}, vol.~6, no.~2, pp. 144--159, 2015.

\bibitem{Hosseini2007}
M.~Hosseini, D.~T. Ahmed, S.~Shirmohammadi, and N.~D. Georganas, ``A survey of
  application-layer multicast protocols,'' \emph{Commun. Surveys Tuts.},
  vol.~9, no.~3, Jul. 2007.

\bibitem{Andreev11}
K.~Andreev, B.~M. Maggs, A.~Meyerson, J.~Saks, and R.~K. Sitaraman,
  ``Algorithms for constructing overlay networks for live streaming,''
  \emph{CoRR}, vol. abs/1109.4114, 2011.

\bibitem{zhou2015joint}
F.~Zhou, J.~Liu, G.~Simon, and R.~Boutaba, ``Joint optimization for the
  delivery of multiple video channels in telco-cdns,'' \emph{{IEEE transactions
  on network and service management}}, vol.~12, no.~1, pp. 87--100, 2015.

\bibitem{fressancourt2016kumori}
A.~Fressancourt, C.~Pelsser, and M.~Gagnaire, ``{Kumori: Steering Cloud traffic
  at IXPs to improve resiliency},'' in \emph{{Proc. of DRCN}}, 2016.

\bibitem{mathis1997macroscopic}
M.~Mathis, J.~Semke, J.~Mahdavi, and T.~Ott, ``The macroscopic behavior of the
  tcp congestion avoidance algorithm,'' \emph{ACM SIGCOMM CCR}, vol.~27, no.~3,
  pp. 67--82, 1997.

\bibitem{even1975complexity}
S.~Even, A.~Itai, and A.~Shamir, ``On the complexity of time table and
  multi-commodity flow problems,'' in \emph{Foundations of Computer Science,
  1975., 16th Annual Symposium on}.\hskip 1em plus 0.5em minus 0.4em\relax
  IEEE, 1975, pp. 184--193.

\bibitem{NET:NET3230140406}
J.~Desrosiers, F.~Soumis, and M.~Desrochers, ``Routing with time windows by
  column generation,'' \emph{Networks}, vol.~14, no.~4, pp. 545--565, 1984.

\bibitem{XiThXu05}
Y.~Xiao, K.~Thulasiraman, and G.~Xue, ``{GEN-LARAC: A generalized approach to
  the constrained shortest path problem under multiple additive constraints},''
  in \emph{Proc. of ISAAC}, 2005, pp. 92--105.

\bibitem{byrka2010improved}
J.~Byrka, F.~Grandoni, T.~Rothvo{\ss}, and L.~Sanit{\`a}, ``{An improved
  LP-based approximation for Steiner tree},'' in \emph{{Proc. of ACM STOC}},
  2010.

\end{thebibliography}

\end{document}